\documentclass{article}
\usepackage{spconf,amsmath,amssymb,amsfonts,epsfig,epstopdf,float,amsthm}
\usepackage{textcomp}

\newtheorem{theorem}{Theorem}

\newtheorem{proposition}[theorem]{Proposition}

\usepackage{delarray}

\newcommand{\ol}{\overline}

\newcommand {\beq} {\begin{equation}}
\newcommand {\eeq} {\end{equation}}
\newcommand {\barr} {\begin{array}}
\newcommand {\earr} {\end{array}}
\newcommand {\bear} {\begin{eqnarray}}
\newcommand {\eear} {\end{eqnarray}}
\newcommand {\bears} {\begin{eqnarray*}}
\newcommand {\eears} {\end{eqnarray*}}

\title{Long-Term Energy Constraints and Power Control in Cognitive
Radio Networks}

\name{Fran\c cois M\'eriaux$^1$, Yezekael Hayel$^2$, Samson Lasaulce$^1$, Andrey Garnaev$^3$}
\address{$^1$L2S - CNRS - SUPELEC - Univ Paris-Sud\\
    F-91192 Gif-sur-Yvette, France\\
    \{meriaux,lasaulce\}@lss.supelec.fr \vspace{0.3cm}\\
    $^2$Lab. d'Informatique d'Avignon - Universit\'e d\textquotesingle Avignon\\
		84911 Avignon - France\\
		yezekael.hayel@univ-avignon.fr\vspace{0.3cm}\\
$^3$V.~I.Zubov Research Institute of Computational Mathematics $\&$ Control Processes \\
St Petersburg State University, Russia 198504\\
agarnaev@rambler.ru}

\begin{document}
\ninept
\maketitle

\begin{abstract}
When a long-term energy constraint is imposed to a transmitter, the average energy-efficiency of a transmitter is, in general, not maximized by always transmitting. In a cognitive radio context, this means that a secondary link can re-exploit the non-used time-slots. In the case where the secondary link is imposed to generate no interference on the primary link, a relevant issue is therefore to know the fraction of time-slots available to the secondary transmitter, depending on the system parameters. On the other hand, if the secondary transmitter is modeled as a selfish and free player choosing its power control policy to maximize its average energy-efficiency, resulting primary and secondary signals are not necessarily orthogonal and studying the corresponding Stackelberg game is relevant to know the outcome of this interactive situation in terms of power control policies.
%about 100 to 150 words, and should be identical to the abstract text submitted electronically

\end{abstract}

%==========================================================

\section{Introduction}

One of the ideas of cognitive radio is to allow some wireless terminals, especially transmitters, to sense their environment in terms of used spectrum and to react to it dynamically. The cognitive radio paradigm~\cite{mitola-1999} has become more and more important to the wireless community since the release of the FCC report~\cite{fcc-report-2002}. Indeed, cognitive radio corresponds to a good way of tackling the crucial problem of spectrum congestion and increasing spectral efficiency. More recently, the main actors of the telecoms industry, namely carriers, manufacturers, and regulators have also realized the importance of energy aspects in wireless networks (see e.g.,~\cite{palicot-icwmc-2009}) both at the network infrastructure and mobile terminal sides. There are many reasons for this and we will not provide them here. As far as this paper is concerned, the goal is to study the influence of long-term energy constraints (e.g., the limited battery life typically) on power control in networks where cognitive radios are involved. The performance criterion which is considered for the terminal is derived from the one introduced by Goodman and Mandayam in~\cite{goodman-pc-2000}. Therein, the authors propose a distributed power control scheme for frequency non-selective block fading multiple access channels. For each block, a terminal aims at maximizing its individual energy-efficiency namely, the number of successfully decoded bits at the receiver per Joule consumed at the transmitter.  Although, a power control maximizing such a performance metric is called energy-efficient, it does not take into account possible long-term energy constraints. Indeed, in~\cite{goodman-pc-2000} and related references (e.g.,~\cite{meshkati-jsac-2006}\cite{lasaulce-twc-2009}), the terminals always transmit, which amounts to considering no constraints on the available (average) energy.  The goal of the present work is precisely to see how energy constraints modify power control policies in a single-user channel and in a cognitive radio channel.  For the sake of simplicity, time-slotted communications are assumed.

The paper is organized in two main parts. In Sec.~\ref{sec:pcs} a single-user channel is considered. It is shown that maximizing an average energy-efficiency under a long-term energy constraint leads the terminal to not transmit on certain blocks.  The probability that the terminal does not transmit is lower bounded. In a setting where a primary transmitter has to control its power under energy-constraint, this probability matters since it corresponds to the fraction of available time-slots which are re-exploitable by a secondary (cognitive) transmitter.  In Sec.~\ref{sec:pcs}, the single-user channel model is sufficient since the secondary link has to meet a zero interference constraint (it can only exploit non-used time-slots). In Sec.~\ref{sec:compet}, the secondary transmitter is assumed to be free to use all the time-slots. The technical difference between the primary and secondary transmitters is that the former has to choose its power level in the first place while the
  latter observes this level and react to it.  The suited interaction model is therefore a Stackelberg game~\cite{stackelberg-book-1934} where the primary and secondary transmitters are respectively the leader and follower of the game. Sec.~\ref{sec:numeric} provides numerical results which allow us to validate some derived results and compare the two cognitive settings (depending whether the secondary transmitter can generate non-orthogonal signals).

%The contributions of this papers are as follows. In Sec.~\ref{sec:pcs}, we introduce a wireless transmission model which performances are based on an energy-efficient metric~\cite{goodman-pc-2000} and which is limited by an energy constraint. When computing the optimal power control scheme in the single-user case, it occurs that the user must not transmit on every time-slots. Thus there are free time-slots available for an hypothtical second user. Then, we consider several scenarios for the two users case: In Sec.~\ref{sec:primary}, the first user, called primary user is prioritary and chooses its power control scheme without taking into account the secondary user and this latter may only transmit when the primary user does not. We show that it is possible to formulate a tight lower-bound to the probability for the primary user not to transmit. In Sec.~\ref{sec:compet}, we express the context in a two-player game and we give the systems correspounding to the resolution of th
 %e Nash equilibrium~\cite{Nash51} and the Stackelberg equilibrium~\cite{stackelberg-book-1934}. In Sec.~\ref{sec:numerical}, we compute numerical simulations to compare and discuss the previous cases.

\newpage
%=============================================================
%=============================================================
\section{General system model}
\label{sec:general-signal-model}

In the whole paper the goal is to study a system comprising two transmitter-receiver
pairs. The signal model under consideration can be described by a frequency non-selective block fading channel. The signals received by the two receivers write as:
\begin{equation} \label{eq:ic-signal-model}
\begin{array}{ccc}
y_1 & = & h_{11} x_1 + h_{21} x_2 + z_1\\
y_2 & = & h_{22} x_2 + h_{12} x_1 + z_2
\end{array}.
\end{equation}
The channel gain of the link $ij$ namely, $h_{ij}$ is assumed to be constant over each block or time-slot. The quantity $g_{ij}= |h_{ij}|^2$ is assumed to be a continuous random variable having independent realizations and distributed according to the probability density function $\phi_{ij}(g_{ij})$. The reception noises are zero-mean complex white Gaussian noises with variance $\sigma^2$. The instantaneous power of the
transmitted signal $x_i$ on time-slot $t$ is given by
\begin{equation}
\label{eq:def-transmit-power}
p_i(t)= \frac{1}{N} \sum_{n=1}^N
|x(n)|^2
\end{equation}
where $n$ is the symbol index and $N$ the number of symbols per time-slot. For simplicity, transmissions are assumed to be time-slotted.

Transmitter $1$ (resp. $2$), receiver $1$ (resp. $2$),
link $11$ (resp. $22$) will be respectively called primary (resp. secondary) transmitter, primary (resp. secondary) receiver, and (resp. secondary) primary
link. The main technical difference between the primary and the secondary links is
that the secondary transmitter can observe the power levels chosen by the primary transmitter but the converse does not hold. In this paper, two scenarios are investigated:
\begin{itemize}
  \item Scenario 1 (Sec. \ref{sec:pcs}): the secondary transmitter is imposed to meet a zero-interference constraint on the primary link.
Since the primary and secondary signals are orthogonal, everything happens for the transmitter $1$ as if it was transmitting over a single-user channel.

  \item Scenario 2 (Sec. \ref{sec:compet}): this time, the secondary transmitter can use all the time-slots and not only those not exploited by the primary link. Primary and secondary signals are therefore not orthogonal in general. In this framework, for each time-slot, the primary transmitter chooses its power level and is informed that the secondary will observe and react to it in a rational manner. A Stackelberg game formulation is proposed to study this interactive situation.
\end{itemize}

%=============================================================
%=============================================================
\section{When primary and secondary signals are orthogonal}
\label{sec:pcs}

%-------------------------------------------------------------
\subsection{Optimal power control scheme for the primary transmitter}
\label{sec:optimal-pc-primary}

From the primary point of view, there is no interference and the signal-to-noise plus
interference ratio (SINR) coincides with the signal-to-noise ratio (SNR):
\begin{equation}
\label{eq:snr} \mathrm{SNR}(p_1(g_{11})) = \frac{g_{11} p_1(g_{11})}{\sigma^2}.
\end{equation}
When using the notation $p_1(g_{11})$ instead of $p_{11}(t)$ we implicitly make appropriate ergodicity assumptions on $g_{11}$. The main purpose of this section is precisely to determine the optimal control function $p_1(g_{11})$ in the sense of the long-term energy efficiency, which is defined as follows:
\begin{equation}
\label{eq:utility} u_1(p_1(g_{11})) = R_1 \int_{0}^{+\infty} \phi_{11}(g_{11})
\frac{f(\mathrm{SNR}(p_1(g_{11})) )}{p_1(g_{11})} \mathrm{d}g_{11}
\end{equation}
where $R_1$ is the transmission rate and $f$ is an efficiency function representing the packet success rate $f:\mathbb{R}^+ \rightarrow [0,1]$. The function $f$ is assumed to possess the following properties:
\begin{enumerate}
\item $f$ is non-decreasing, C$^2$ differentiable, $f(0)=0$, $\lim\limits_{x \to +\infty}f(x) = 1$ and there exists a unique inflection point $x_0$ for $f$.
\item $f'$ is non-negative, $f'(0) = \lim\limits_{x \to +\infty}f'(x) = 0$. $f'$ reaches its maximum for $x_0$.
\item $f''$ is non-negative over $[0,x_0]$, negative over $[x_0,+\infty[$. $f^{(2)}(0)=0$, $\lim\limits_{x \to +\infty}f''(x) = 0^{-}$.
\end{enumerate}
These properties are verified by the two typical efficiency functions available
 in the literature:
\begin{equation}
f_a(x)=\begin{array}\{{cc}.
e^{-\frac ax}\; \ \forall x > 0 \\
0\; \text{if }x=0
\end{array}
\end{equation}
and
\begin{equation}
f_M(x)=\left(1-e^{-x}\right)^M\; \forall x  \geq 0.
\end{equation}
The function $f_a$, $a\geq0$ has been introduced in~\cite{belmega-valuetools-2009}\cite{belmega-tsp-2010} and corresponds to the case where the efficiency function equals one minus the outage probability. On the other hand, $f_M$, $M \in \mathbb{N}^*$, corresponds to an empirical approximation of the packet success rate which was already used in~\cite{goodman-pc-2000}.

Compared to references~\cite{goodman-pc-2000}\cite{meshkati-jsac-2006}\cite{lasaulce-twc-2009}, note that the user's utility is the average energy-efficiency and not the instantaneous energy-efficiency. This allows one to take into account the following energy constraint:
\begin{equation}
\label{eq:energy-constraint} T \int_{0}^{+\infty} \phi_{11}(g_{11})  p_1(g_{11})
\mathrm{d}g_{11} \leq E_1
\end{equation}
where $T$ is the time-slot duration and $E_1$ is the available energy for terminal $1$. In order to
find the optimal solution(s) for the power control schemes, let us consider the Lagrangian $L_{u_1}$. It writes as:

\begin{equation}
\begin{aligned}
L_{u_1} &= R_1 \int_{0}^{+\infty} \phi_{11}(g_{11})\frac{f(\mathrm{SNR}(p_{1}(g_{11})) )}{p_1(g_{11})} \mathrm{d}g_{11}\\
&- \lambda( T \int_{0}^{+\infty} \phi_{11}(g_{11})  p_1(g_{11})
\mathrm{d}g_{11} - E_1).
\end{aligned}
\end{equation}

It is ready to show that the optimal instantaneous signal-to-noise
ratio (\ref{eq:snr}) has to be the solution of $\frac{\partial L_{u_1}}{\partial p_1(g_{11})} = 0$:
\begin{equation}
\label{eq:gamma-star} x f'(x) - f(x) = \frac{\lambda T \sigma^4 }{R_1 g_{11}^2} x^2.
\end{equation}
%where $\lambda$ is the Lagrangian associated with
%(\ref{eq:energy-constraint}).
Solving the above equation amounts to finding the zeros of $F(x) = x f'(x) - f(x) - \frac{\lambda T \sigma^4 }{R_1 g_{11}^2} x^2$. We have that
$F$ is C$^1$ differentiable, $F(0)=0$, $\lim\limits_{x \to +\infty}F(x) = -\infty$, and
\begin{equation}
F'(x)=x f^{(2)}(x) - 2 \frac{\lambda T \sigma^4 }{R_1 g_{11}^2} x.
\end{equation}
Then, $\exists \epsilon, \; \forall x \in ]0,\epsilon],\; F'(x) < 0$.
%$F'(0) = 0$, $\forall \left(u(n)\right)_{n \in \N} | \lim_{n \to +\infty} u(n)= 0$, $\exists n_0 | \forall n \geq n_0, F'(u(n))<0$.
Considering the sign of $F'$, given the particular form of $f^{(2)}$, two cases have to be considered.
\begin{itemize}
\item If $\forall x$, $f''(x) \leq 2 \frac{\lambda T \sigma^4 }{R_1 g_{11}^2}$, $F'$ is negative or null and $F$ is decreasing. Then $0$ is the only zero for $F$.
\item If $\exists (x_1,x_2),\; x_1<x_2$ st $f''(x_1) = f''(x_2) = \frac{\lambda T \sigma^4 }{R_1 g_{11}^2}$, and $F'$ non-negative over $[x_1,x_2]$. $F$ decreases over $[0,x_1]$, increases over $[x_1,x_2]$ and decreases over $[x_2,+\infty[$. Then $F$ may have zero, one or two zeros different from $0$.
\end{itemize}
%\end{itemize}

%\begin{equation}
%\begin{aligned}
%&F(0) = 0 \\
%&F'(x) = x f^{(2)} - 2 \frac{\lambda T \sigma^4 }{R g^2} x
%\end{aligned}
%\end{equation}

%\begin{figure}
%\begin{center}
%%\hspace{-1cm}
%\includegraphics[scale=0.6]{derivsecond}
%\caption{Typical form of $f^{(2)}(x)$. In this particular case $f(x) = e^{-\frac ax}$ and $a=3$.}
%\label{Fig:derivsecond}
%\end{center}
%\end{figure}

%\begin{figure}
%\begin{center}
%%\hspace{-1cm}
%\includegraphics[scale=0.5]{fillustration}
%\caption{Typical form of $f(x)$ and its two first derivatives. In this particular case $f(x) = e^{-\frac ax}$ and $a=10$.}
%\label{Fig:fillustration}
%\end{center}
%\end{figure}

%Considering the typical form of $f^{(2)}$ for a sigmoid function (see Fig.~\ref{Fig:derivsecond}), there are two cases we can consider
%\begin{itemize}
%\item $\forall x$, $f^{(2)}(x) \leq 2 \frac{\lambda T \sigma^4 }{R g^2}$. In this case, $F'$ is negative or null and $F$ is decreasing. $0$ is the only zero for $F$
%\item $\exists (x_1,x_2),\; x_1<x_2$ st $f^{(2)}(x_1) = f^{(2)}(x_2) = \frac{\lambda T \sigma^4 }{R g^2}$. In this case, $F$ decreases over $[0,x_1]$, increases over $[x_1,x_2]$ and decreases over $[x_2,+\infty[$. $F$ may have zero, one or two zeros different from $0$.
%\end{itemize}

If $F$ has one zero, it is $0$ and $0$ is the maximum for $L_{u_1}$. If $F$ has two zeros: $0$ and $x'_0$, $L_{u_1}$ is decreasing and $0$ is the maximum for $L_{u_1}$. If $F$ has three zeros: $0$, $x'_1$ and $x'_2$, $L_{u_1}$ decreases over $[0,x'_1]$, increases over $[x'_1,x'_2]$ and decreases over $[x'_2,+\infty[$. The maximum for $L_{u_1}$ is then $0$ or $x_2$.
%Fig.~\ref{Fig:FSNR_A} and ~\ref{Fig:FSNR_M} illustrates the three possible cases for functions $e^{-\frac ax}$ and $(1-e^{-x})^M$.

%\begin{figure}
%\begin{center}
%%\hspace{-1cm}
%\includegraphics[scale=0.5]{FSNR_M}
%\caption{Several profiles of $F$ for $f(x)=(1-e^{-x})^M$.}
%\label{Fig:FSNR_M}
%\end{center}
%\end{figure}

%For $f(x) = e^{-\frac{a}{x}}$, (\ref{eq:gamma-star}) turns into
%\begin{equation}
%(a-x)e^{-\frac{a}{x}} = \frac{\lambda T \sigma^4}{R g^2} x^3.
%\end{equation}
%For $f(x) = (1-e^{-x})^M$, (\ref{eq:gamma-star}) turns into
%\begin{equation}
%(1-e^{-x})^{M-1}(e^{-x}(M x +1) -1) = \frac{\lambda T \sigma^4}{R g^2} x^2.
%\end{equation}

%Again, the solutions of the above equation should be \textbf{studied
%in details}, depending on the efficiency function $f$ and the
%parameters $\lambda, T, ...$.

%Fig.~\ref{Fig:SNR} illustrates function $F$ with $f(x) = e^{-\frac ax}$ for various values of $g^2$. We can see that there are positive solutions of (\ref{eq:gamma-star}) if $g^2$ is good enough and that these solutions are not unique. In Fig.~\ref{Fig:SNR}, it is clearly viewable for $g^2 = 10^{-11}$.

%\begin{figure}
%\begin{center}
%%\hspace{-1cm}
%\includegraphics[scale=0.6]{otpimalSNR}
%\caption{Function F depending on SNR.}
%\label{Fig:SNR}
%\end{center}
%\end{figure}

%\begin{figure}
%\begin{center}
%%\hspace{-1cm}
%\includegraphics[scale=0.5]{counter}
%\caption{In this example, $\mathrm{snr}_{\lambda_E}^*(g)=0$ but $g > \sigma^2 \sqrt{2 \frac{\lambda T }{R \max_{x}f^{(2)}(x)}}$}
%\label{Fig:Counter}
%\end{center}
%\end{figure}

Assume $\mathrm{SNR}_{\lambda_{E_1}}^*(g)$ is the greatest solution of equation (\ref{eq:gamma-star}).
%Note that we have $\mathrm{snr}_{\lambda_E}^*(g)=0$ if and only if $g \leq \sigma^2 \sqrt{2 \frac{\lambda T }{R \max_{x}f^{(2)}(x)}}$. AAA Francois: this is not true, we can only write $\mathrm{snr}_{\lambda_E}^*(g)=0$ if $g \leq \sigma^2 \sqrt{2 \frac{\lambda T }{R \max_{x}f^{(2)}(x)}}$, see Fig.~\ref{Fig:Counter} for counter-example.ZZZ
Then an optimal power control scheme is given by:
\begin{equation}
\label{eq:pc-scheme}
p_1^*(g_{11}) = \frac{\sigma^2}{g_{11}} \mathrm{SNR}_{\lambda_{E_1}}^*(g_{11})
\end{equation}
with $\mathrm{SNR}_{\lambda_{E_1}}^*(g_{11}) \geq 0 $. Since $E_1$ is fixed,
the methodology consists in determining $\lambda_{E_1} $,
%: \textbf{Is it always optimal to saturate the energy-constraint?},
then a solution
of (\ref{eq:gamma-star}) is determined numerically. Note that $\lambda_{E_1}$ is in bits/Joule$^2$. It can be interpreted as a minimal number of bits to transmit for $1$ Joule$^2$. The higher $\lambda_{E_1}$ is, the better the channel should be to be used.

\textbf{Remark (Capacity of fast fading channels).} The proposed analysis is reminiscent
to the capacity determination of fast fading single-user channels~\cite{goldsmith-tit-1997}. Two important
differences between this and our analysis are worth being emphasized. First, mathematically, the optimization problem under study is more general than the one of~\cite{goldsmith-tit-1997}. Indeed, if one makes the particular choice $f(\mathrm{SNR}(p_{1}(g_{11}))) = p_{1} \log\left(
1 +  \mathrm{SNR}(p_{1}(g_{11}))\right)$, the optimal SNR is given by $\mathrm{SNR}^*(p_1(g_{11})) =
\frac{g_{11}}{\lambda_{E_1} \sigma^2} - 1$, which corresponds to a water-filling solution (the SNR has to be non-negative). Second, the physical interpretation of the average utility is different from the fast fading case. In the fast fading case, the power control is updated at the symbol rate whereas in our case, it is updated at the time-slot frequency namely, $\frac{1}{T}$. Indeed, in power control problems, what is updated is the average power over a block or time-slot and assuming an average power constraint over several blocks or time-slots generally does not make sense. However, from an energy perspective introducing an average constraint is relevant. This comment is a kind of subtle and characterizes our approach.

%\emph{Example:} \emph{the ergodic Shannon rate for fast fading
%channels. Then}
%\begin{equation}
%\frac{f(\mathrm{snr}(p(g)) )}{p(g)} = \log\left(1+
%\mathrm{snr}(p(g)) \right).
%\end{equation}
%\emph{It can be checked that: the energy constraint has to be
%saturated, the solution is unique, $\mathrm{snr}^*(p(g)) =
%\frac{g}{\lambda \sigma^2} - 1$. The non-negativeness constraint
%leads to the classical water-filling solution [goldsmith-it-1997]:}
%\begin{equation}
%p^*(g) =  \left[ \frac{\sigma^2}{g} \mathrm{snr}_{\lambda_E}^*(g)
%\right]^+= \left[ \frac{1}{\lambda_E} - \frac{\sigma^2}{g}\right]^+.
%\end{equation}
%\emph{One important point to notice is that, depending on $g$, the
%transmitter will transmit or not transmit. This feature, which is
%assumed to be also available for other choices of $f$, is exploited
%in the following section.}

%==========================================================
\subsection{Time-slot occupancy probability}
\label{sec:primary}

As shown in the preceding section, time-slots are not used by the primary link when
the solution $\mathrm{SNR}_{\lambda_{E_1}}^*(g_{11})$ is negative. Therefore, the probability that
this event occurs corresponds to the probability of having a free time-slot for
the secondary link. It is thus relevant to evaluate $\mathrm{Pr}[\mathrm{SNR}_{\lambda_{E_1}}^*(g_{11}) \leq 0]$. At first glance, explicating this probability does not seem to be trivial. However,
one can see from the preceding section that if $\max f'' \leq 2 \frac{\lambda T \sigma^4 }{R_1 g_{11}^2}$, the function $F$ has no non-negative solutions except from $0$, in which case there is no power allocated to channel $g_{11}$. Based on this observation, the following lower bound arises:
\begin{equation}
\mathrm{Pr}\left[\max f'' \leq 2 \frac{\lambda T \sigma^4 }{R_1 g_{11}^2}\right] \leq \mathrm{Pr}[\mathrm{SNR}_{\lambda_{E_1}}^*(g_{11}) \leq 0].
\end{equation}
Many simulations have shown that this lower bound is reasonably tight, 
one of them is provided in the
simulation section; what matters in this paper is to show that 
the fraction of available time-slots can be significant and the proposed
 lower bound ensures to achieve at least the corresponding performance. To
  conclude on this point, note that in the case where $f(\mathrm{SNR}(p_{1}(g_{11}))) = p_{1} \log\left(
1 +  \mathrm{SNR}(p_{1}(g_{11}))\right)$, the probability of having a free time-slot for the secondary link can be easily expressed and is given by:
\begin{equation}
\mathrm{Pr}\left[\mathrm{SNR}_{\lambda_{E_1}}^*(g_{11}) \leq 0\right] = 1 - e^{- \frac{\lambda_{E_1}
\sigma^2}{\ol{g}_{11}}}
\end{equation}
where $\ol{g}_{11} = E(g_{11})$. A similar analysis has been made to design a
Shannon-rate efficient interference alignment technique for static MIMO interference channels~\cite{medina-pimrc-2008}\cite{medina-tsp-submitted}.

%=============================================================
%=============================================================
\section{A Stackelberg formulation of the non-orthogonal case}
\label{sec:compet}

We assume now that both transmitters are free to decide their power control
policy. However, there is still hierarchy in the system in the sense that, for each time-slot, the primary transmitter has to choose its power level in the first place and the secondary transmitter (assumed to equipped with a cognitive radio) observes this level and reacts to it. This framework is
exactly the one of a Stackelberg game since it is assumed that the primary transmitter (called the game leader) knows it is observed by a rational player (the game follower). The SINR for the first transmitter/receiver pair is:
%$$
%SINR_1(p_1(g_{11}),p_2(g_{21}))=\frac{p_1(g_{11})}{\sigma^2+p_{2}(g_{21})}:=\gamma_1,
%$$
\begin{equation}
SINR_1(p_1,p_2)=\frac{p_1 g_{11}}{\sigma^2+p_{2} g_{21}}:=\gamma_1,
\end{equation}
where $g_{21}$ is the channel gain between transmitter 2 and receiver 1. For the second transmitter/receiver pair, the SINR is:
%$$
%SINR_2(p_1(g_{12}),p_2(g_{22}))=\frac{p_2(g_{22})}{\sigma^2+p_{1}(g_{12})}:=\gamma_2,
%$$
\begin{equation}
SINR_2(p_1,p_2)=\frac{p_2 g_{22}}{\sigma^2+p_{1} g_{12}}:=\gamma_2,
\end{equation}
where $g_{12}$ is the channel gain between transmitter 1 and receiver 2.
Using this relation, we have the powers for transmitters 1 and 2 depending on the SINRs:
\begin{equation}
\begin{aligned}
&p_1=\frac{\sigma^2}{g_{11}}\frac{\gamma_1+\gamma_1\gamma_2\frac{g_{21}}{g_{22}}}{1-\alpha \gamma_1\gamma_2},\quad \mbox{and} \quad
p_2=\frac{\sigma^2}{g_{22}}\frac{\gamma_2+\gamma_1\gamma_2\frac{g_{12}}{g_{11}}}{1-\alpha \gamma_1\gamma_2} \\
&\mbox{with} \quad \alpha=\frac{g_{21}g_{12}}{g_{11}g_{22}}.
\end{aligned}
\end{equation}

%We are considering a particular competition framework in which the secondary transmitter/receiver pair will choose his power control scheme after observing the power control scheme of the primary transmitter/receiver pair. Actually, the primary user knows that the secondary transmitter/receiver pair will use this logic and then, the primary transmitter/receiver pair determines his power control scheme considering the best power control scheme of the secondary transmitter/receiver pair. That type of hierarchical game is a Stackelberg game and has been recently used to study optimal power control in wireless networks \cite{HLH11}. This framework gave interesting results which can be extended to our long term power control scheme.

A Stackelberg equilibrium is a vector $(p_1^*,p_2^*)$ such that:
\begin{equation}
p_1^* = \arg\max_{p_1}u_1(p_1,p_2^*(p_1)),
\end{equation}
with
\begin{equation}
\forall p_1, \quad p_2^*(p_1)=\arg\max_{p_2}u_2(p_1,p_2).
\end{equation}
%A Stackelberg equilibrium is not necessary unique and it is easy to check that a Nash equilibrium is a %Stackelberg equilibrium.
Note that the above expression implicitly assumes that the best-response of the follower is a singleton,
which is effectively the case for the problem under study. In our Stackelberg game, the utility $u_2$ of the secondary transmitter/receiver pair depends on the power control scheme $p_1$ through the expression:
%$$
%\forall p_1,\quad u_2(p_1,p_2)=R \int_{0}^{+\infty}\int_{0}^{+\infty} \phi_1(g_{12})\phi_2(g_{22})\frac{f(\frac{p_2g_{22}}{\sigma^2+p_{1}g_{12}})}{p_2} \mathrm{d}g_{12}\mathrm{d}g_{22},
%$$
\begin{equation}
\begin{aligned}
&\forall p_1,\quad u_2(p_1,p_2)= \\
&R_2 \int_{0}^{+\infty}\int_{0}^{+\infty} \phi_{12}(g_{12})\phi_{22}(g_{22})\frac{f(\frac{p_2 g_{22}}{\sigma^2+p_{1}g_{12}})}{p_2} \mathrm{d}g_{12}\mathrm{d}g_{22},
\end{aligned}
\end{equation}
with the energy constraint:
\begin{equation}
T \int_{0}^{+\infty} \phi_{22}(g_{22})  p_2
\mathrm{d}g_{22} \leq E_2.
\end{equation}
In order to determine a Stackelberg equilibrium, we first have to express the best response of the follower that is, the best power control scheme for the secondary transmitter/receiver pair, given the long term power control scheme of the primary transmitter/receiver pair.

For a given $p_1(g_{12})$, the Lagrangian $L_{u_2}$ of $u_2$ is given by:

\begin{equation}
\begin{aligned}
&L_{u_2}(p_1,p_2,\lambda_2)= \\
&R_2 \int_{0}^{+\infty}\int_{0}^{+\infty} \phi_{12}(g_{12})\phi_{22}(g_{22})\frac{f(\frac{p_2 g_{22}}{\sigma^2+p_{1}g_{12}})}{p_2} \mathrm{d}g_{12}\mathrm{d}g_{22} \\
&- \lambda_2( T \int_{0}^{+\infty} \phi_{22}(g_{22})  p_2\mathrm{d}g_{22} - E_2).
\end{aligned}
\end{equation}

\begin{proposition}[Optimal SINR for the secondary transmitter] The secondary transmitter has to tune
its power level such that its SINR is the greatest zero of the following equation:
\begin{equation}
 x f'(x) - f(x) = \frac{\lambda_2 T (\sigma^2+p_1 g_{12})^2 }{R_2 g_{22}^2} x^2.
\label{sinr2stack}
\end{equation}
\end{proposition}
The proof is ready and follows the single-user case analysis, which is conducted in Sec. \ref{sec:pcs}. The optimal power control scheme $p_2^*$ of the secondary transmitter/receiver pair, depending on the power control scheme $p_1$ is given by:
%$$
%p_2^*(p_1)=\frac{\sigma^2+p_1g_{12}}{g_{22}}x_2(p_1(g_{12})),
%$$
\begin{equation}
p_2^*(p_1)=\frac{\sigma^2+p_1 g_{12}}{g_{22}}x_2(p_1),
\end{equation}
where $x_2(p_1)$ is the greatest solution of (\ref{sinr2stack}).

Now, let us the consider the case of the primary transmitter.

\begin{proposition}[Optimal SINR for the primary transmitter] The primary transmitter has to tune
its power level such that its SINR is the greatest zero of the following equation:
\scriptsize
\begin{equation}
\begin{aligned}
&x f'(x)\left[1-\alpha x_2x-G(x)\right] - f(x) = \frac{\lambda_1 T \sigma^4}{R_1 g_{11}^2} \left(\frac{1+\frac{g_{21}}{g_{22}}x_2}{1-\alpha x x_2}\right)^2x^2,\\
&\mbox{with} \quad G(x) = \frac{\alpha x(1+\frac{g_{12}}{g_{11}}x)^2x_2}{(1-\alpha x_2x)^2\frac{R_2 g_{22}^2}{2\lambda_2 T\sigma^4}f''(x_2)-(1+\frac{g_{12}}{g_{11}}x)^2}.
\label{sinr1stack}
\end{aligned}
\end{equation}
\normalsize
\end{proposition}

\begin{proof}
The leader is optimizing his utility function $u_1$ taking into account this best response power control scheme of the follower transmitter/receiver pair. The SINR of the leader transmitter/receiver pair, when the follower transmitter/receiver pair uses his best response power control scheme, is given by:
%$$
%SINR_1(p_1(g_{11}))=\frac{p_1(g_{11})g_{11}}{\sigma^2+p^*_2(p_1(g_{11}))g_{21}}=\frac{p_1(g_{11})g_{11}}{\sigma^2(1+x_2(p_1(g_{11})))+p_1(g_{11})g_{11}x_2(p_1(g_{11}))}.
%$$
\begin{equation}
\begin{aligned}
\mathrm{SINR}_1(p_1,p^*_2(p_1))&=\frac{p_1 g_{11}}{\sigma^2+p^*_2(p_1)g_{21}}\\
&=\frac{p_1 g_{11}}{\sigma^2(1+\frac{g_{21}}{g_{22}}x_2(p_1))+p_1 \frac{g_{12}g_{21}}{g_{22}}x_2(p_1)}.
\end{aligned}
\end{equation}
The derivative of the SINR of the leader is
\begin{equation}
\begin{aligned}
&\frac{\partial \gamma_1}{\partial p_1}(p_1)=\\
&g_{11}\frac{\sigma^2(1+\frac{g_{21}}{g_{22}}x_2(p_1))-p_1\sigma^2\frac{g_{21}}{g_{22}}x'_2(p_1)-p_1^2\frac{g_{12}g_{21}}{g_{22}}x'_2(p_1)}{(\sigma^2(1+\frac{g_{21}}{g_{22}}x_2(p_1))+p_1 \frac{g_{12}g_{21}}{g_{22}}x_2(p_1))^2}
\end{aligned}
\end{equation}
Then we have
\scriptsize
\begin{equation}
\begin{aligned}
&p_1\frac{\partial \gamma_1}{\partial p_1}(p_1)=\\
&\gamma_1(p_1)\frac{\sigma^2(1+\frac{g_{21}}{g_{22}}x_2(p_1))-p_1\sigma^2\frac{g_{21}}{g_{22}}x'_2(p_1)-p_1^2\frac{g_{12}g_{21}}{g_{22}}x'_2(p_1)}{\sigma^2(1+\frac{g_{21}}{g_{22}}x_2(p_1))+p_1 \frac{g_{12}g_{21}}{g_{22}}x_2(p_1)},\\
&=\gamma_1(p_1)\left(1- \frac{p_1 \frac{g_{12}g_{21}}{g_{22}}x_2(p_1)+p_1\sigma^2\frac{g_{21}}{g_{22}}x'_2(p_1)+p_1^2\frac{g_{12}g_{21}}{g_{22}}x'_2(p_1)}{\sigma^2(1+\frac{g_{21}}{g_{22}}x_2(p_1))+p_1 \frac{g_{12}g_{21}}{g_{22}}x_2(p_1)}\right),\\
&=\gamma_1(p_1)\left(1-x_2(p_1)\alpha \gamma_1(p_1)- \frac{(\sigma^2+p_1g_{12})p_1\frac{g_{21}}{g_{22}}x'_2(p_1)}{\sigma^2(1+\frac{g_{21}}{g_{22}}x_2(p_1))+p_1 \frac{g_{12}g_{21}}{g_{22}}x_2(p_1)}\right),\\
&=\gamma_1(p_1)\left(1-x_2(p_1)\alpha \gamma_1(p_1)-\frac{\sigma^2+p_1g_{12}}{g_{12}}\alpha x'_2(p_1)\gamma_1(p_1)\right)
\end{aligned}
\end{equation}
\normalsize

Taking the expression of $x_2(p_1)$ we get:
\begin{equation}
\begin{aligned}
&x_2' f'(x_2)+x_2x'_2f''(x_2) - x'_2f'(x_2) = \\
&2\frac{\lambda_2 T (\sigma^2+p_1g_{12}) }{R_2 g_{22}^2} g_{12}x_2^2+2\frac{\lambda_2 T (\sigma^2+p_1g_{12})^2 }{R_2 g_{22}^2} x_2x'_2,
\end{aligned}
\end{equation}
which yields to:
\begin{equation}
x'_2f''(x_2)=2\frac{\lambda_2 T (\sigma^2+p_1g_{12}) }{R_2 g_{22}^2} g_{12}x_2+2\frac{\lambda_2 T (\sigma^2+p_1g_{12})^2 }{R_2 g_{22}^2} x'_2.
\end{equation}
Then we get the derivative of $x_2(p_1)$:
\begin{equation}
x'_2(p_1)=\frac{\frac{2\lambda_2 T}{R_2 g_{22}^2}(\sigma^2+p_1g_{12})g_{12}x_2}{f''(x_2)-\frac{2\lambda_2 T}{R_2 g_{22}^2}(\sigma^2+p_1g_{12})^2}.
\end{equation}
Then we have:
\begin{equation}
\frac{(\sigma^2+p_1g_{12})x'_2(p_1)}{g_{12}}=\frac{\frac{2\lambda_2 T}{R_2 g_{22}^2}(\sigma^2+p_1g_{12})^2x_2}{f''(x_2)-\frac{2\lambda_2 T}{R_2 g_{22}^2}(\sigma^2+p_1g_{12})^2}.
\end{equation}
Taking the expression of the power of receiver/transmitter pair 1 depending on both SINRs, we get:
\begin{equation}
\sigma^2+p_1g_{12}=\sigma^2\left(\frac{1+\frac{g_{12}}{g_{11}}\gamma_1}{1-\alpha \gamma_1 \gamma_2}\right),
\end{equation}
Then
\small
\begin{equation}
\frac{(\sigma^2+p_1g_{12})x'_2(p_1)}{g_{12}}=\frac{(1+\frac{g_{12}}{g_{11}}\gamma_1)^2x_2}{(1-\alpha \gamma_2 \gamma_1)^2\frac{R_2 g_{22}^2}{2\lambda_2 T\sigma^4}f''(x_2)-(1+\frac{g_{12}}{g_{11}}\gamma_1)^2}.
\end{equation}
\normalsize

Then we have:
\small
\begin{equation}
\begin{aligned}
&p_1\frac{\partial \gamma_1}{\partial p_1}(p_1)=\\
&\gamma_1\left(1-\alpha x_2 \gamma_1-\frac{\alpha \gamma_1(1+\frac{g_{12}}{g_{11}}\gamma_1)^2x_2}{(1-\alpha x_2 \gamma_1)^2\frac{R g_{22}^2}{2\lambda_2 T\sigma^4}f''(x_2)-(1+\frac{g_{12}}{g_{11}}\gamma_1)^2}\right)
\end{aligned}
\end{equation}
\normalsize
\end{proof}

By denoting $x_1$ the largest solution of this equation, the optimal power control scheme of the leader at the equilibrium is given by:
%\frac{p_1^*(g_{11})g_{11}}{\sigma^2(1+x_2(p_1^*(g_{11})))+p_1^*(g_{11})g_{11}x_2(p_1^*(g_{11}))}=x_1.
%$$
\begin{equation}
\frac{p_1^* g_{11}}{\sigma^2(1+\frac{g_{21}}{g_{22}}x_2(p_1^*))+p_1 \frac{g_{12}g_{21}}{g_{22}}x_2(p_1^*)}=x_1.
\end{equation}

\normalsize

\section{Numerical Results}
\label{sec:numeric}

%Fig.~\ref{Fig:FSNR_A} illustrates the possible profiles of function $F$ for $f=f_1$.
%\begin{figure}
%\begin{center}
%%\hspace{-1cm}
%\includegraphics[scale=0.5]{FSNR_A}
%\caption{Several profiles of $F$ for $f_1(x)$.}
%\label{Fig:FSNR_A}
%\end{center}
%\end{figure}

The following simulations are performed with the parameters: $T=10^{-3}$ s, $R_1=R_2=10^4$ bits/s, $\sigma^2=10^{-12}$ W, the channel gains $g_{11}$ and $g_{22}$ are assumed to follow a Rayleigh distribution of mean $10^{-10}$, when needed, $g_{12}$ and $g_{21}$ are assumed to follow a Rayleigh distribution of mean $10^{-12}$ and the efficiency function used is $f_a$, defined in Sec. \ref{sec:pcs} with $a=0.9$. Fig.~\ref{Fig:Energy} illustrates the influence of $\lambda_E$ on the energy constraint in a single-user case. When $\lambda_E$ is low, the optimal power control scheme is to transmit most of the time, thus the energy spent is high. On the contrary, when $\lambda_E$ increases, transmission will only occurs when the channel gain is good enough, resulting in a lower energy spent. After a certain threshold, the optimal scheme is not to transmit at all.
\begin{figure}[H]
\begin{center}
%\hspace{-1cm}
\includegraphics[scale=0.6]{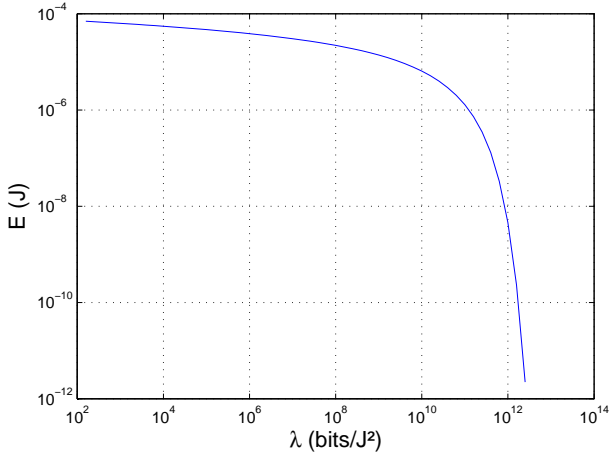}
\caption{Energy spent on duration T depending on $\lambda_E$.}
\label{Fig:Energy}
\end{center}
\end{figure}

In Fig.~\ref{Fig:proba}, we are in the context of Sec. \ref{sec:primary}. We compute the probability per time-slot that the primary link is not used and we compare it to its lower bound. It is interesting to note that this lower-bound is relatively tight to the exact probability.
\begin{figure}
%\begin{center}
\hspace{-0.5cm}
\includegraphics[scale=0.5]{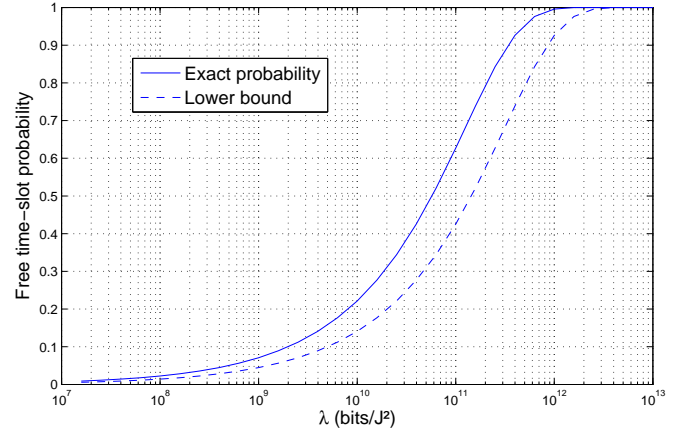}
\caption{Comparison of the exact probability of having free time-slot with the proposed lower bound of this probability.}
\label{Fig:proba}
%\end{center}
\end{figure}

%Fig.~\ref{Fig:instauti} illustrates the instantaneous utility of a player at the Nash equilibrium depending on the channel states. on one side, if the player has a bad channel, he does not transmit. On the other side, if he has a good channel while the other has a bad channel, he reaches his maximum utility. Finally, if the two players have good channel states, they transmit together, leading to a lower instantaneous utility.
%\begin{figure}
%\begin{center}
%%\hspace{-1cm}
%\includegraphics[scale=0.45]{instauti}
%\caption{Instantaneous utility of a player depending on $(g_1,g_2)$. In this case, we choose $\lambda_1 = \lambda_2$.}
%\label{Fig:instauti}
%\end{center}
%\end{figure}

Fig.~\ref{Fig:compauti} compares the expected utilities of Stackelberg equilibrium (Sec. \ref{sec:compet}) and the orthogonal case (Sec. \ref{sec:pcs}). 
%For this simulation, we choose $\lambda_1=\lambda_2=\lambda$, $g_{11}=g_{12}$ and $g_{22}=g_{21}$. Note that in these conditions, the considered interference channel is equivalent to a multiple access channel (MAC). 
As we could expect, the primary link of the orthogonal case offers the best utility, but the orthogonal secondary link has the worst performance. The leader and follower of the Stackelberg case have are much more similar in terms of performance and are very clos to the performance of the primary link which makes the Stackelberg case a very efficient and fair scenario for both links. 
%It is also very interesting to note Stackelberg utilities increase with $\lambda$. We explain this phenomenon by the fact that for small $\lambda$, both players will often transmit together creating a lot of interference. When $\lambda$ gets higher, players transmit less often and then they reduce their probability to interfere. 
Of course, like in the single-user case, after a threshold for $\lambda$, they do not transmit at all.
\begin{figure}[H]
%\begin{center}
\hspace{-0.5cm}
\includegraphics[scale=0.5]{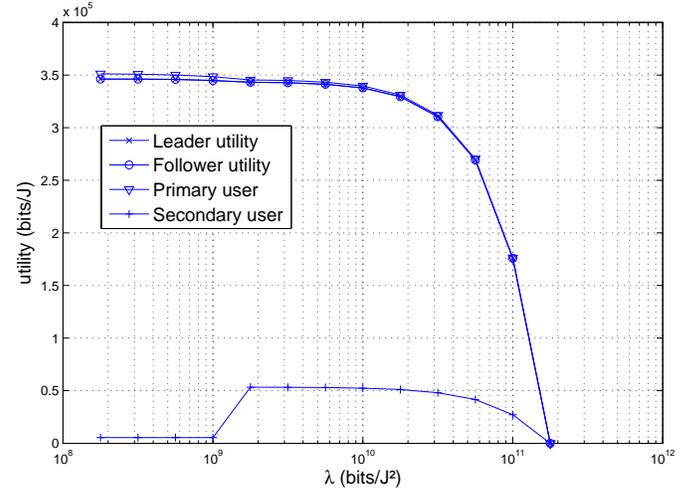}
\caption{Comparison of the expected utilities of Stackelberg equilibrium and the orthogonal case depending on $\lambda$. In this particular case, $\lambda_1=\lambda_2=\lambda$.}
\label{Fig:compauti}
%\end{center}
\end{figure}

In particular, Fig.~\ref{Fig:powerpro} shows the optimal power profile of the leading transmitter w.r.t. the channels gains $g_{11}$ and $g_{22}$ when $\lambda = 10^10$ bits/J$^2$. It is clear that for low values of $g_{11}$, the optimal policy is not to transmit. Then we distinguish two zones of interest:
\begin{itemize}
\item when both $g_{11}$ and $g_{22}$ are good, the transmitter uses most of its power for a relatively high value of $g_{11}$,
\item when only $g_{11}$ is good, we can see that the transmitter uses most of its power for a lower value of $g_{11}$ as it is not likely to facing interference from the following transmitter in this zone.
\end{itemize}

\begin{figure}
\begin{center}
%\hspace{-1cm}
\includegraphics[scale=0.45]{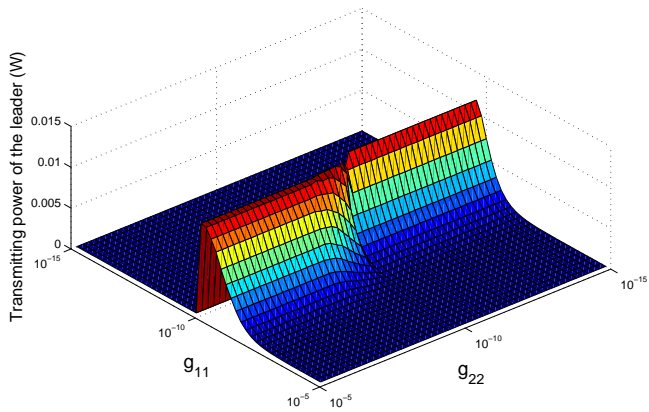}
\caption{Power profile of the leading transmitter w.r.t. $g_{11}$ and $g_{22}$ in the two-player Stackelberg case.}
\label{Fig:powerpro}
\end{center}
\end{figure}

\section{Conclusion and Perspectives}
In this paper, it is shown how a long-term energy constraint modifies
the behavior of a transmitter in terms of power control policy. In contrast with related
works such as~\cite{goodman-pc-2000}\cite{meshkati-jsac-2006}\cite{lasaulce-twc-2009}, a transmitter does
not always transmit when it is subject to such a constraint. This shows that when implementing its best power control policy, a primary link does not exploit all the available time-slots. The probability of having a free time-slot for the secondary link can be lower bounded in a reasonably tight manner and shown to be non-negligible in general. As a second step, a scenario where the secondary link can interfere on the primary link is analyzed. The problem is formulated as a Stackelberg game where the primary transmitter is the leader and the secondary transmitter is the follower. An equilibrium in this game is shown to exist for typical conditions on the efficiency function $f(x)$. Interestingly, the fact that the transmitters have a long-term energy constraint can make the system more efficient since this incites users to interfere less; indeed simulations show the existence of a value of an energy budget which maximizes the users's utilities. While the power control schemes at the equilibrium can be determined, the corresponding equations have a drawback: the power control scheme of a given user does not only rely on the knowledge of its individual channel gain but also on the other channel gains. This shows the relevance of improving the proposed work by designing more distributed power control policies. Additionally, the proposed scenarios included one primary link and one secondary link. When several cognitive transmitters are present, there is a competition between the secondary transmitters for exploiting the resources left by the primary link.

\newpage

\bibliographystyle{IEEEbib}
\bibliography{biblio-book-2011-03-13}

\end{document}